\documentclass[a4paper,11pt,reqno]{article}

\usepackage[T1]{fontenc}
\usepackage[latin1]{inputenc}

\usepackage{setspace}
\usepackage{indentfirst}

\usepackage{geometry}
\usepackage{hyperref}

\usepackage{amsrefs}

\usepackage[intlimits]{amsmath}
\usepackage{amssymb}
\usepackage{mathrsfs}
\usepackage{stmaryrd}

\usepackage{amsthm}

\theoremstyle{plain}
\newtheorem{prop}{Proposition}[section]
\newtheorem{lem}[prop]{Lemma}

\newtheorem{cor}[prop]{Corollary}
\newtheorem{thm}[prop]{Theorem}

\newtheorem*{prop*}{Proposition}
\newtheorem*{lem*}{Lemma}
\newtheorem*{sublem*}{Sublemma}
\newtheorem*{cor*}{Corollaire}
\newtheorem*{thm*}{Theorem}
\newtheorem*{hypo*}{Hypothesis}
\newtheorem*{question*}{Question}
\newtheorem*{conjecture*}{Conjecture}
\newtheorem*{scholum*}{Scholum}
\newtheorem{defn}[prop]{Definition}
\newtheorem*{defn*}{Definition}

\newtheoremstyle{slanted}
  {3pt}
  {3pt}
  {\slshape}
  {}
  {\bfseries}
  {.}
  {.5em}
  {}

\theoremstyle{slanted}
\newtheorem{ex}[prop]{Example}

\newtheorem*{ex*}{Example}
\newtheorem*{exs*}{Examples}
\newtheorem{rmk}[prop]{Remark}

\newtheorem*{rmk*}{Remark}
\newtheorem*{rmks*}{Remarks}

\newtheorem*{notation*}{Notation}

\theoremstyle{definition}

\newtheorem*{con*}{Construction}
\newtheorem*{note*}{Note}

\theoremstyle{remark}

\newtheorem*{warning*}{Warning}
\newtheorem*{shortnote*}{Note}
\newtheorem*{claim*}{Claim}
\newtheorem*{axiom*}{Axiom}

\newcommand{\CC}{\mathbb{C}}
\newcommand{\reals}{\mathbb{R}}

\newcommand{\ZZ}{\mathbb{Z}}

\newcommand{\thalf}{\tfrac{1}{2}}

\newcommand{\smalcirc}{\mbox{\tiny{$\circ$}}}

\newcommand{\ba}[2]{[#1,#2]}
\newcommand{\bas}[2]{[#1,#2]_*}

\newcommand{\lon }{\,\rightarrow\,}
\newcommand{\abs}[1]{\left\vert#1\right\vert}

\newcommand{\pairing}[2]{\left\langle #1  |  #2 \right\rangle}

\newcommand{\inserts}{\iota}

\newcommand{\set}[1]{\left\{#1\right\}}

\newcommand{\Poissonbracket}[1]{\left \{ #1\right \}}

\newcommand{\thetaalgebra}{\mathfrak{\theta}}
\newcommand{\abelianideal}{\mathfrak{I}}
\newcommand{\thetaalgebrastar}{\mathfrak{\theta}^*}
\newcommand{\galgebra}{\mathfrak{g}}
\newcommand{\galgebrastar}{\mathfrak{g}^*}

\newcommand{\crossedmodulealgebra}{\galgebra\ltimes  \thetaalgebra}

\newcommand{\unit}{{\bf{1}}}

\newcommand{\adjoint}[1]{\mathrm{ad}_{#1}}
\newcommand{\coadjoint}[1]{\mathrm{ad}^*_{#1}}

\newcommand{\phiUprotate }{\phi^{\scriptscriptstyle T}}
\newcommand{\Id}{{\bf{1}}}

\newcommand{\ddelta}{    { \delta}}

\newcommand{\skewdelta}{{\delta}}
\newcommand{\domega}{  { \omega}}
\newcommand{\dpi}{  { \pi}}

\newcommand{\crossedmoduletriple}[3]{(#1\stackrel{#2}{\rightarrow}#3)}
\newcommand{\Lietwo}[2]{ {#2 \ltimes #1}}

\newcommand{\moduleaction}{ \triangleright}

\newcommand{\phipush}{D_{\phi}}

\newcommand{\minuspower}[1]{(-1)^{#1}}

\newcommand{\thetaalgebradegone}{\mathrm{\theta}}
\newcommand{\thetaalgebrastardegminustwo}{\mathrm{\theta}^*}
\newcommand{\galgebradegzero}{\mathfrak{g}}
\newcommand{\galgebrastardegminusone}{\mathfrak{g}^*}
\newcommand{\phimap}{\phi}
\newcommand{\bracketmapbracket}[2]{[#1,#2]}
\newcommand{\actionmapaction}[2]{#1 \succ #2}
\newcommand{\homotopymap}{h}

\newcommand{\cohomotopymap}{\tilde{\eta}}

\newcommand{\symmetricproduct}{{\scriptstyle \odot}\,}

\newcommand{\symmetricalgebra}{\mathcal{S}^{\bullet}}
\newcommand{\derivedby}[1]{ {D}_{#1}}
\newcommand{\Sbullet}{S^\bullet}

\newcommand{\tobefilledin}{\,\stackrel{\centerdot}{}\,}
\newcommand{\Linf}{L_{\infty}}
\newcommand{\degreesubspace}[2]{{#1}^{(#2)}}

\newcommand{\Vs}{V^*}
\newcommand{\US}{\mathfrak{S}}

\newcommand{\Koszul}[1]{\epsilon(#1)}

\newcommand{\shiftby}[2]{{#1}{\scriptstyle{[#2]}}}
\newcommand{\ipG}[2]{(#1,#2)^{\galgebra}}
\newcommand{\plankconstant}{\hbar}

\newcommand{\derivesh}{\varepsilon^{\scriptscriptstyle 03}_{\scriptscriptstyle 10}}
\newcommand{\derivesaction}{\varepsilon^{\scriptscriptstyle 01}_{\scriptscriptstyle 11}}
\newcommand{\derivesbracket}{\varepsilon^{\scriptscriptstyle 12}_{\scriptscriptstyle 00}}
\newcommand{\derivesvarphi}{\varepsilon^{\scriptscriptstyle  10}_{\scriptscriptstyle 01}}
\newcommand{\derivescoh}{{\varepsilon^{\scriptscriptstyle 01}_{\scriptscriptstyle 30}}}
\newcommand{\derivescoaction}{\varepsilon^{\scriptscriptstyle 11}_{\scriptscriptstyle 10}}
\newcommand{\derivescobracket}{\varepsilon^{\scriptscriptstyle 00}_{\scriptscriptstyle 21}}
\newcommand{\derivescovarphi}{\varepsilon^{\scriptscriptstyle 10}_{\scriptscriptstyle 01}}
\newcommand{\CP}{\mathrm{~c.p.~}}
\newcommand{\el}{l}

\hypersetup{pdfpagemode=UseOutlines,colorlinks=false,pdfpagelayout=SinglePage,pdfstartview=FitH,bookmarksopen=true}

\begin{document}

\title{Weak Lie 2-Bialgebras\footnote{Research partially supported by NSF grants DMS-0605725, DMS-0801129, DMS-1101827 and NSFC grant 11001146.}}
\author{\textsc{Zhuo Chen} \\
{\small Department of Mathematics, Tsinghua University} \\
{\small
\href{mailto:zchen@mail.math.tsinghua.edu.cn}
{\texttt{zchen@math.tsinghua.edu.cn}}}\and
\textsc{Mathieu Sti\'enon} \\
{\small Department of Mathematics, Pennsylvania State University} \\
{\small
{\href{mailto:stienon@math.psu.edu}{\texttt{stienon@math.psu.edu}}}}
\and \textsc{Ping Xu}
\\
{\small Department of Mathematics, Pennsylvania State University} \\
{\small
 {\href{mailto:ping@math.psu.edu}{\texttt{ping@math.psu.edu}}}
 } }

\date{}

\maketitle
\begin{abstract}
We introduce the notion of weak Lie 2-bialgebra.
Roughly, a weak Lie 2-bialgebra is a pair of compatible 2-term $L_\infty$-algebra structures
on a vector space and its dual. The compatibility condition is described in terms of the big bracket. We prove that (strict) Lie 2-bialgebras are in one-one correspondence
with crossed modules of Lie bialgebras.
\end{abstract}

\tableofcontents

\section{Introduction}

The main purpose of the paper is to develop the notion of weak Lie 2-bialgebras.
A Lie bialgebra is a Lie algebra endowed with a compatible Lie coalgebra structure.
Lie bialgebras can be regarded as the classical limits of quantum groups.
A celebrated theorem of Drinfeld establishes a bijection between Lie bialgebras
and connected, simply connected Poisson Lie groups.
Poisson 2-groups \cite{CSX} are a natural first step in the search for an appropriate notion
of quantum 2-groups, which can be considered as deformation quantization of 
ordinary  Lie  2-groups. Their infinitesimal counterparts are called Lie 2-bialgebras 
or crossed modules of Lie bialgebras.


Recall that a Lie algebra crossed module consists of a pair of Lie algebras
$\thetaalgebra$ and $\galgebra$ together with a linear map
$\phi:~\thetaalgebra\to\galgebra$ and an action of $\galgebra$ on
$\thetaalgebra$ by derivations satisfying a certain compatibility condition.
A Lie bialgebra crossed module  is  a pair of Lie algebra crossed modules
in duality: $\crossedmoduletriple{\thetaalgebra}{\phi}{\galgebra}$ and
$\crossedmoduletriple{\galgebrastar}{-\phi^*}{\thetaalgebrastar}$
are both Lie algebra crossed modules, and
$(\Lietwo{\thetaalgebra}{\galgebra},\Lietwo{\galgebrastar}{\thetaalgebrastar})$
is a Lie bialgebra.

It is well known that Lie algebra crossed modules are a special
case of weak Lie 2-algebras (i.e.\ two-term $L_\infty$ algebras) \cite{MR2068522}.
It is natural to ask  what is a weak Lie 2-bialgebra.
Such an object ought to be a weak Lie 2-algebra as well as a weak Lie
2-coalgebra, both structures being compatible with one another in a certain sense.
Several notions of $L_\infty$ bialgebras can be found in the existing literature,
among which we can mention
Kravchenko's homotopy Lie bialgebras \cite{MR2327020} and Merkulov's
homotopy Lie[1] bialgebras \cite{MR2600029}.
However, none of them serves our purpose.
Although it is a weak Lie 2-algebra, a two-term homotopy Lie bialgebra in the sense of Kravchenko is, for instance, not a weak Lie 2-coalgebra due to the degree convention   {(see Remark \ref{rmk:different notions})}.
To obtain the correct compatibility condition, it turns out that one must
shift the degree on the underlying $\ZZ$-graded vector space $V$ so as to modify the ``big bracket'', which is a Gerstenhaber bracket on ${\Sbullet}(\shiftby{V}{2}\oplus
{\shiftby{\Vs}{1}})$.
Identifying ${\Sbullet}(\shiftby{V}{2}\oplus{\shiftby{\Vs}{1}})$
with the space $\Gamma_{ }(\wedge^\bullet\shiftby{T}{4}M)$
of  polyvector fields on $M=\shiftby{\Vs}{-2}$,
the big bracket can be simply described as the Schouten bracket of multivector fields.
\emph{In terms of the big bracket, a weak Lie 2-bialgebra on a graded vector space $V$
is a degree $(-4)$ element $\varepsilon$ of ${\Sbullet}(\shiftby{V}{2}\oplus{\shiftby{\Vs}{1}})$
such that $\{\varepsilon, \varepsilon\}=0$.}
 Lie 2-bialgebras arise as a special case of weak Lie 2-bialgbras where certain homotopy terms vanish.
Our main theorem establishes a bijection between   Lie 2-bialgebras
and crossed modules of Lie bialgebras.

This is the first of a series of papers devoted to the study
of Poisson 2-groups \cite{CSX} and their quantization.
We are grateful to the organizers of ``Journ\'ee Quantique'' (June 2010),
``WAGP 2010'' (June 2010), and ``Poisson 2010'' (July 2010),
where we had the pleasure to present our results.
Since then drafts of this work and slides of our conference talks have circulated in the community.  
{Some of the results presented here were reproduced in an arXiv preprint posted in 2011 \cite{BSZ}.}
We would like to thank several institutions for their hospitality
while work on this project was being done: Penn State University (Chen),
Universit\'e du Luxembourg (Chen,  Sti\'enon and Xu),
Institut des Hautes \'Etudes Scientifiques
and Beijing International Center for Mathematical Research (Xu).
We would also like to thank Anton Alekseev, Benjamin Enriques, 
Yvette Kosmann-Schwarzbach, Henrik Strohmayer, Jim Stasheff, and Alan 
Weinstein for useful discussions and comments.  
We are grateful to the anonymous referee for carefully reading this paper.

Some notations are in order.

\textsc{Notations:}  {In this paper, all vector spaces are assumed to be finite dimensional.} Given a graded vector space $V=\bigoplus_{k\in\ZZ}V^{(k)}$, $V[i]$ denotes the graded
vector space obtained by shifting the grading on $V$ according to the rule
$(V[i])^{(k)}=V^{(i+k)}$, and $V^*$ denotes the dual vector space, which is graded according to the rule $(V^*)^{(-k)}=(V^{(k)})^*$. Note in particular that $(V[i])^*=(V^*)[-i]$.
We write $\abs{e}$ for the degree of a homogeneous vector $e\in V$.
The symbol $\symmetricproduct$ is used for the symmetric tensor product: for any homogeneous vectors $e,f\in V$,
\[ e\symmetricproduct f=\tfrac{1}{2}(e\otimes f+(-1)^{|e| |f|}  f\otimes e) .\]
The symmetric algebra over $V$ will be denoted by ${\Sbullet}(V)$.

\section{Lie 2-bialgebras}
\subsection{The big  bracket}
\label{Schoutenbracket}

We will introduce a graded version  of the  big bracket \cites{MR2103012, MR1046522}
involving graded vector spaces.

Let $V=\bigoplus_{k\in\ZZ} \degreesubspace{V}{k}$ be a $\ZZ$-graded vector
space.  Consider the   $\ZZ$-graded manifold  $M=\shiftby{\Vs}{-2}$ and
the shifted  tangent space
\[ \shiftby{T}{4}M\cong \shiftby{(M\times \shiftby{\Vs}{-2})}{4}\cong M\times{\shiftby{\Vs}{2}} .\]

Consider the space of polyvector fields on $M$ with polynomial coefficients:
\begin{eqnarray*}  \Gamma_{ }(\wedge^\bullet
\shiftby{T}{4}M)\cong \Sbullet(M^*)\otimes   \Sbullet  (\shiftby{(\shiftby{\Vs}{2})}{-1}) \cong
 \Sbullet(\shiftby{V}{2})\otimes
\Sbullet(\shiftby{\Vs}{1})\cong  {\Sbullet}(\shiftby{V}{2}\oplus
{\shiftby{\Vs}{1}}).
\end{eqnarray*} 
{Here, elements of $\shiftby{\Vs}{-2}$, when thought as having degree 1  more than their actual degrees in $\shiftby{\Vs}{-2}$, are exactly the sections of $TM$ which are constant along $M$, i.e., the vector fields on $M$ invariant under translation.}

In the sequel, let us denote
${\Sbullet}(\shiftby{V}{2}\oplus
{\shiftby{\Vs}{1}})$ by $\symmetricalgebra$. The symmetric tensor product on
 $\symmetricalgebra$ will be denoted by $\symmetricproduct$.

There is a standard way to endow $\symmetricalgebra=\Gamma_{ }(\wedge^\bullet\shiftby{T}{4}M)$ with a graded Lie bracket, i.e. the Schouten bracket,
 denoted by $\Poissonbracket{\tobefilledin,\tobefilledin}$. It is  a
 bilinear map
 $\Poissonbracket{\tobefilledin,\tobefilledin}:$  $\symmetricalgebra\otimes
 \symmetricalgebra\lon \symmetricalgebra$ satisfying the following properties:
\begin{itemize}
\item[1)] $\Poissonbracket{v,v'}=
    \Poissonbracket{\epsilon,\epsilon'}=0$, $\Poissonbracket{v,\epsilon}= \minuspower{\abs{v}}\pairing{v}{\epsilon}
$, $\forall$ $v, v'\in \shiftby{V}{2}$,
$\epsilon,\epsilon'\in \shiftby{\Vs}{1}$;
\item[2)]{ $\Poissonbracket{e_1,e_2}=-\minuspower{ (\abs{e_1}+3)
(\abs{e_2}+3)  }\Poissonbracket{e_2,e_1} $, $\forall$
$e_i\in \symmetricalgebra$; }
\item[3)]$\Poissonbracket{e_1,e_2\symmetricproduct e_3}=
\Poissonbracket{e_1,e_2}\symmetricproduct e_3 + \minuspower{(
\abs{e_1}+3)\abs{e_2}}e_2\symmetricproduct
\Poissonbracket{e_1,e_3}$,  $\forall$
$e_i\in \symmetricalgebra$.
\end{itemize}

It is clear that $\Poissonbracket{\tobefilledin,\tobefilledin}$
is of degree $3$, i.e. $$\abs{\Poissonbracket{e_1,e_2}}=\abs{e_1} +\abs{e_2}+3,$$
for all homogeneous $e_i\in \symmetricalgebra$, 
and the following graded Jacobi identity holds:
$$\Poissonbracket{e_1,\Poissonbracket{e_2,e_3}}=\Poissonbracket{
\Poissonbracket{e_1,e_2},e_3}
+\minuspower{(\abs{e_1}+3)(
\abs{e_2}+3)}\Poissonbracket{e_2,\Poissonbracket{e_1,e_3}}  .
$$

Hence $(\symmetricalgebra,\symmetricproduct,
\Poissonbracket{\tobefilledin,\tobefilledin})$ is a Schouten algebra, also known as an odd Poisson algebra, or a Gerstenharber algebra \cite{MR1958834}.

\begin{rmk}
Due to our degree convention, when $V$ is a vector space considered
as a graded vector space concentrated at degree $0$,
the big bracket above  is different from the usual
big bracket in the literature \cite{MR2103012}.
\end{rmk}
{
An element $F\in S^p(\shiftby{V}{2})\symmetricproduct
S^q(\shiftby{\Vs}{1})$ can be considered as a $q$-polyvector field
on $M=\shiftby{\Vs}{-2}$, while an  element $x\in {{\Sbullet}(\shiftby{V}{2})}$ can be considered as a function on $\shiftby{\Vs}{-2}$. Therefore, by applying $F$ to $q$-tuples of functions on $M=\shiftby{\Vs}{-2}$, i.e., taking
successive  Schouten brackets of $F$ with functions on $M=\shiftby{\Vs}{-2}$, we obtain a multilinear map:
\begin{align*}
&\derivedby{F}:
     &        \underbrace{ {{\Sbullet}(\shiftby{V}{2})} )\otimes \cdots \otimes {{\Sbullet}(\shiftby{V}{2})}} &
            ~\longrightarrow~  {{\Sbullet}(\shiftby{V}{2})},&
\\
&   &  \mbox{ $q$ -tuples\qquad ~~~ }
        &&&
\end{align*}
by
$$\qquad\derivedby{F}(x_1,\cdots,x_q)
=\Poissonbracket{
            \Poissonbracket{\cdots
                        \Poissonbracket{
                            \Poissonbracket{F,x_1},x_2},
                    \cdots
                ,x_{q-1}
                },                        x_q
            },$$
for all $x_i\in {{\Sbullet}(\shiftby{V}{2})}$.
}

It is easy to see that
\begin{equation}\label{Eqt:derivedbyentriesalternate}
\derivedby{F}(x_1,\cdots, x_i,x_{i+1},\cdots,x_q)
=(-1)^{(\abs{x_i}+3)(\abs{x_{i+1}}+3)}\derivedby{F}(x_1,\cdots,
x_{i+1},x_i,\cdots,x_q).
\end{equation}

For any $E\in S^k(\shiftby{V}{2})\symmetricproduct
S^l(\shiftby{\Vs}{1})$, and  $F\in S^p(\shiftby{V}{2})\symmetricproduct
S^q(\shiftby{\Vs}{1})$, we have
\begin{eqnarray}
&&\derivedby{\Poissonbracket{E,F}}(x_1,\cdots,x_{n}) \label{Eq:bigbracket}\\
&=&\sum_{\sigma\in \US(q,l-1)}\Koszul{\breve{\sigma}}\derivedby{E}
    (
        \derivedby{F}(x_{\sigma(1)},\cdots,x_{\sigma(q)}),
        x_{\sigma(q+1)},\cdots,x_{\sigma(n)}
    ) \nonumber
\\&&\quad
    -\minuspower{(\abs{E}+3)(\abs{F}+3)}
    \sum_{\sigma\in
\US(l,q-1)}\Koszul{\breve{\sigma}}\derivedby{F}
    (
        \derivedby{E}(x_{\sigma(1)},\cdots,x_{\sigma(l)}),
        x_{\sigma(l+1)},\cdots,x_{\sigma(n)}
    ), \nonumber
\end{eqnarray}
for all $x_1,\cdots,x_n\in {{\Sbullet}(\shiftby{V}{2})}$,
where $n=q+l-1$.
Here  $\US(j,n-j)$ denotes the collection of
$(j,n-j)$-shuffles and $\Koszul{\breve{\sigma}}$ denotes
the Koszul sign: switching any two successive elements
 $x_{i}$ and $x_{i+1}$  leads to
 a sign change $\minuspower{(\abs{x_i}+3)(\abs{x_{i+1}}+3)}$.

\subsection{Weak Lie 2-algebras, coalgebras and bialgebras}\label{subsec:weakLie2algebras}

Following Baez-Crans \cite{MR2068522},
 {
there is an equivalence of 2-categories between the 2-category of weak Lie 2-algebras  the 2-category of 2-term $\Linf$-algebras.}
Unfolding the $\Linf$-structure on the 2-term graded vector
space $V=\theta\oplus \galgebradegzero$, where $\theta$ is of degree $1$ and $\galgebradegzero$ is of degree $0$, we can equivalently define
a weak Lie 2-algebra as a
 pair of vector spaces $\thetaalgebradegone$ and
$\galgebradegzero$ endowed with the following structures:

\begin{enumerate}
\item a linear map $\phimap$: $\thetaalgebradegone\lon \galgebradegzero$;
\item a bilinear skew-symmetric
 map  $\bracketmapbracket{\tobefilledin}{\tobefilledin}$:
 $\galgebradegzero\otimes\galgebradegzero\to
\galgebradegzero$;
\item a bilinear map  $\actionmapaction{\tobefilledin}{\tobefilledin}$:
$\galgebradegzero\otimes\thetaalgebradegone\lon\thetaalgebradegone$;
\item a trilinear skew-symmetric map $\homotopymap$: $\galgebradegzero\otimes\galgebradegzero\otimes\galgebradegzero\lon
\thetaalgebradegone$,  called the homotopy map.
\end{enumerate}

These maps are  required to satisfy  the following compatibility conditions:
for all $w,x,y,z \in \galgebradegzero$ and $u,v\in
\thetaalgebradegone$,
\begin{enumerate}
\item[]  \begin{equation}\bracketmapbracket{
\bracketmapbracket{x}{y}}{z}+ \bracketmapbracket{
\bracketmapbracket{y}{z}}{x}+\bracketmapbracket{
\bracketmapbracket{z}{x}}{y}+(\phimap\circ\homotopymap)(x,y,z)=0;
\end{equation}
\item[]
\begin{equation}
\actionmapaction{y}{(\actionmapaction{x}{u})}
-\actionmapaction{x}{(\actionmapaction{y}{u})} +\actionmapaction{
\bracketmapbracket{x}{y}}{u}+\homotopymap(\phimap(u),x,y)=0;
\end{equation}
\item[]
\begin{equation}
\actionmapaction{\phimap(u)}{v}+\actionmapaction{\phimap(v)}{u}=0;
\end{equation}
\item[]
\begin{equation}\phimap(\actionmapaction{x}{u})=\bracketmapbracket{x}{\phimap(u)};
\end{equation}
\item[]\begin{eqnarray}\nonumber
&&
-\actionmapaction{w}{\homotopymap(x,y,z)}
-\actionmapaction{y}{\homotopymap(x,z,w)}
+\actionmapaction{z}{\homotopymap(x,y,w)}
+\actionmapaction{x}{\homotopymap(y,z,w)}
\\\nonumber &=&
\homotopymap(\bracketmapbracket{x}{y},z,w)-\homotopymap(\bracketmapbracket{x}{z},y,w)
+\homotopymap(\bracketmapbracket{x}{w},y,z)+\homotopymap(\bracketmapbracket{y}{z},x,w)\\
&&
-\homotopymap(\bracketmapbracket{y}{w},x,z)+\homotopymap(\bracketmapbracket{z}{w},x,y).
\end{eqnarray}
\end{enumerate}

If $\homotopymap$ vanishes, we call it a strict Lie 2-algebra, or simply a Lie 2-algebra.

Now consider the degree-shifted vector spaces $\shiftby{V}{2}$ and
$\shiftby{\Vs}{1}$.
Under such a degree convention, the degrees of $\galgebradegzero$,
$\thetaalgebradegone$, $\galgebrastardegminusone$ and
$\thetaalgebrastardegminustwo$ are specified in the following table:
\begin{center}
    \begin{tabular}{   l  | l |  l | l | l }
     space & $\galgebradegzero$ & $\thetaalgebradegone$ & $\galgebrastardegminusone$
     & $\thetaalgebrastardegminustwo$ \\\hline
   degree & $-2$&$-1$&$-1$& $-2$
\end{tabular}
\end{center}
{
\begin{rmk}The reason that we are using such a  degree convention  can be summarized as follows. First,   under such assumptions, the elements  $\el$ in Proposition \ref{Prop:Lie2algebraelements}, $c$ in Proposition \ref{Prop:Lie2coalgebraelements} and $\varepsilon$ in Definition \ref{Defn:Lie2bialgebra} all will be  of homogenously  degree $(-4)$. Second, we see from the above table that $(\thetaalgebra,\galgebra)$ and $(\galgebrastar,\thetaalgebrastar)$ are symmetric. In fact, when we define the notion of a Lie bialgebra crossed module in the sequel, we are asking $\crossedmoduletriple{\thetaalgebra}{\phi}{\galgebra} $ and
$\crossedmoduletriple{\galgebrastar}{-\phi^*
}{\thetaalgebrastar} $ both to be Lie algebra crossed modules.
\end{rmk}}
We will maintain this convention throughout this paper.
Recall that $\Sbullet=\Sbullet({\shiftby{\Vs}{1}\oplus \shiftby{V}{2}})$.

\begin{prop}
\label{Prop:Lie2algebraelements}
Under the above degree convention, a weak Lie 2-algebra
structure is equivalent to a solution to the equation:
\begin{equation}\label{ll=0}
\Poissonbracket{\el,\el}=0,
\end{equation}
where $\el=\derivesvarphi+\derivesbracket+
\derivesaction+\derivesh$ in $\degreesubspace{\mathcal{S}}{-4}$ such that
\begin{equation}\label{eqt:derivesdata}\left\{
\begin{array}{r@{~}l}
 & \derivesvarphi\in \thetaalgebrastardegminustwo \symmetricproduct
\galgebradegzero,\\
& \derivesbracket\in (\symmetricproduct^2\galgebrastardegminusone)\symmetricproduct \galgebradegzero,\\
& \derivesaction\in
\galgebrastardegminusone\symmetricproduct\thetaalgebrastardegminustwo\symmetricproduct\thetaalgebradegone,\\
& \derivesh\in
(\symmetricproduct^3\galgebrastardegminusone)\symmetricproduct
\thetaalgebradegone.
\end{array}
\right.
\end{equation}
Here the bracket in Eq.\eqref{ll=0} stands for the big bracket as in Section \ref{Schoutenbracket},  {and the notation $${\varepsilon^{\scriptscriptstyle  pq }_{\scriptscriptstyle kl}}\in (\symmetricproduct^q\galgebrastardegminusone)\symmetricproduct(\symmetricproduct^l\thetaalgebrastardegminustwo)\symmetricproduct(\symmetricproduct^k\thetaalgebradegone)\symmetricproduct (\symmetricproduct^p\galgebra)$$ helps the reader to keep track of its underlying space}.
\end{prop}
\begin{proof}

There is a bijection between
the structure maps $\phimap$, $\bracketmapbracket{\tobefilledin}{\tobefilledin}$, $\actionmapaction{\tobefilledin}{\tobefilledin}$ and $\homotopymap$ and the data $\derivesvarphi$, $\derivesbracket$, $\derivesaction$ and $\derivesh$.
They are related by the following equations:
\begin{eqnarray*}
\phimap(u)&=&\derivedby{\derivesvarphi}(u),\\
\bracketmapbracket{x}{y}&=&\derivedby{\derivesbracket}(x,y),\\
\actionmapaction{x}{u}&=&\derivedby{\derivesaction}(x,u),\\
\homotopymap(x,y,z)&=&\derivedby{\derivesh}(x,y,z),
\end{eqnarray*}
$\forall x,y,z\in\galgebra$, $u\in \thetaalgebra$.

 Since $\el=\derivesvarphi+\derivesbracket+
\derivesaction+\derivesh \in \degreesubspace{\mathcal{S}}{-4}$,
 a  simple computation leads to
\begin{eqnarray*}\Poissonbracket{\el,\el} &=&\Poissonbracket{\derivesbracket,\derivesbracket}
 +\Poissonbracket{\derivesaction,\derivesaction}\\&&
 +2\Poissonbracket{\derivesvarphi,\derivesbracket}
 +2\Poissonbracket{\derivesvarphi,\derivesaction}
 +2\Poissonbracket{\derivesvarphi,\derivesh}
+2\Poissonbracket{\derivesbracket,\derivesaction}
 +2\Poissonbracket{\derivesbracket,\derivesh}+2\Poissonbracket{\derivesaction,\derivesh}.
\end{eqnarray*}

By using Eq. \eqref{Eq:bigbracket}, we have, $\forall x,y,z\in\galgebra$,
 \begin{eqnarray}
 \derivedby{\Poissonbracket{\el,\el}}(x,y,z)&=& \derivedby{\Poissonbracket{\derivesbracket,\derivesbracket}}(x,y,z)
 +2\derivedby{\Poissonbracket{\derivesvarphi,\derivesh}}(x,y,z) \nonumber\\
 &=&
2\derivedby{\derivesbracket}(\derivedby{\derivesbracket}(x,y),z)+\CP
+2\derivedby{\derivesvarphi}(\derivedby{\derivesh}(x,y,z) \nonumber
\\
&=&2(\bracketmapbracket{
\bracketmapbracket{x}{y}}{z}+\CP +\phimap\circ\homotopymap(x,y,z)). \label{eq:1}
 \end{eqnarray}

Similarly,
\begin{eqnarray}
 \derivedby{\Poissonbracket{\el,\el}}(x,y,u)
 &=&2 \bigl(\actionmapaction{y}{(\actionmapaction{x}{u})}
-\actionmapaction{x}{(\actionmapaction{y}{u})} +\actionmapaction{ \bracketmapbracket{x}{y}}{u}\bigr)+\homotopymap(\phimap(u),x,y), \label{eq:2} \\
\derivedby{\Poissonbracket{\el,\el}}(u,x)
&=& 2\bigl( \phimap(\actionmapaction{x}{u})+\bracketmapbracket{\phimap(u)}{x}\bigr), \label{eq:3} \\
\derivedby{\Poissonbracket{\el,\el}}(u,v)
&=&2\bigl(\actionmapaction{\phimap(u)}{v}+\actionmapaction{\phimap(v)}{u}\bigr), \label{eq:4}\\
\derivedby{\Poissonbracket{\el,\el}}(x,y,z,w)
&=&2\bigl(\homotopymap(\bracketmapbracket{x}{y},z,w)+\actionmapaction{w}{\homotopymap(x,y,z)}+\CP\bigr)  \label{eq:5}.
\end{eqnarray}

It thus follows that $\Poissonbracket{\el,\el}$ vanishes if and only if
the LHS of Eqs. \eqref{eq:1}-\eqref{eq:5} vanish.
The latter is equivalent to the compatibility conditions defining a weak Lie 2-algebra.
This concludes the proof.
\end{proof}

In the sequel, we denote a weak Lie 2-algebra by
$(\thetaalgebradegone{{\lon}}\galgebradegzero,l)$ in order to emphasize
the map from $\thetaalgebradegone$ to $\galgebradegzero$.
Sometimes, we will omit $l$ and denote a weak Lie 2-algebra simply by
$(\thetaalgebradegone{{\lon}}\galgebradegzero)$.
If $\crossedmoduletriple{\galgebrastardegminusone}{ }{\thetaalgebrastardegminustwo}$
is a weak Lie 2-algebra, then $\crossedmoduletriple{\thetaalgebra}{ }{\galgebra}$
is called a weak Lie 2-coalgebra.

\begin{rmk}
Equivalently, a weak Lie 2-coalgebra underlying $\crossedmoduletriple{\thetaalgebra}{ }{\galgebra}$ is a 2-term $\Linf$-structure on $\galgebrastardegminusone\oplus\thetaalgebrastardegminustwo$,
where $\galgebrastardegminusone$ has degree $1$ and
$\thetaalgebrastardegminustwo$ has degree $0$.
\end{rmk}

Similarly, we have the following

\begin{prop}
\label{Prop:Lie2coalgebraelements}
A weak Lie 2-coalgebra is equivalent to a solution to the equation:
\begin{equation}\nonumber
\Poissonbracket{c,c}=0,
\end{equation}
where  $c=\derivescovarphi+\derivescobracket+ \derivescoaction+\derivescoh\in
\degreesubspace{\mathcal{S}}{-4}$ such that
 \begin{equation}\label{eqt:derivescodata}\left\{
\begin{array}{r@{~}l}
 & \derivescovarphi\in \thetaalgebrastardegminustwo \symmetricproduct
\galgebradegzero,\\
& \derivescobracket\in  \thetaalgebrastardegminustwo \symmetricproduct (\symmetricproduct^2\thetaalgebradegone),\\
& \derivescoaction\in
\galgebrastardegminusone\symmetricproduct\galgebradegzero\symmetricproduct\thetaalgebradegone,\\
& \derivescoh\in
\galgebrastardegminusone\symmetricproduct(\symmetricproduct^3\thetaalgebradegone).
\end{array}
\right.
\end{equation}

\end{prop}

We denote such a weak Lie 2-coalgebra by
$(\thetaalgebradegone{{\lon}}\galgebradegzero,c)$.

Now we are ready to introduce the main object of this section.

\begin{defn}\label{Defn:Lie2bialgebra}
A   weak   Lie 2-bialgebra consists of a pair of vector spaces
 $\theta$ and $\galgebradegzero$ together with
a solution $\varepsilon=  \derivesbracket+
\derivesaction+\derivesh+\derivescovarphi+\derivescobracket+ \derivescoaction+\derivescoh\in
 \degreesubspace{\mathcal{S}}{-4}$
to the equation:
$$\{\varepsilon, \varepsilon\}=0.$$
Here  $\derivesbracket,\derivesaction,\derivesh,  \derivescovarphi,\derivescobracket,\derivescoaction,\derivescoh$
are as in  Eqs. \eqref{eqt:derivesdata} and \eqref{eqt:derivescodata}.

If,  moreover,  $\derivesh=0$, it is called  a quasi-Lie 2-bialgebra.
 If both $\derivesh$ and $\derivescoh$ vanish, we say that the Lie 2-bialgebra
 is strict, or simply a Lie 2-bialgebra.
 \end{defn}

\begin{rmk}\label{rmk:different notions}
   Note that, in the literature, there exist    notions of
homotopy Lie bialgebras \cite{MR2327020} and   Lie 2-bialgebras
\cite{MR2600029}. However, weak Lie 2-bialgebras in our sense are
 neither of them. 
The pattern is that any of these notions are given by 
a homogenous element $h$ in a (even or odd)  Poisson algebra satisfying $\{h,h\}=0$,
 whose bracket $\{\cdot, \cdot\}$ are analogues of  the usual big bracket
 of Kosmann-Schwarzbach \cite{MR2103012}.
The big bracket in \cite{MR2327020} is defined on
 $\Sbullet(V\oplus \Vs)$ (without any degree shifting) and generalizes 
the usual big bracket. Our big bracket in Section \ref{Schoutenbracket},
however, does not reduce to the usual big bracket when $V$ is a vector space
considered as a graded vector space concentrated at degree $0$.
As a result,  the usual Lie bialgebras  in the sense of Drinfeld
 \cite{MR934283} are not  Lie 2-bialgebras in our sense,
but are homotopy Lie bialgebras in the sense of Kravchenko \cite{MR2327020}.

On the other hand, the big bracket in  \cite{MR2600029} is  an 
odd Poisson structure, and  defined on $\Sbullet(V\oplus \shiftby{\Vs}{ 1})$.
 Although it can
 be identified with our big bracket  in Section \ref{Schoutenbracket} under some proper degree adjustment,
 an element $h$  that form  a Lie 2-bialgebra in the sense of Merkulov
 does not define a weak Lie 2-bialgebra in our sense because it has 
 degree $2$ as an element in $\Sbullet(V\oplus \shiftby{\Vs}{1})$, whereas in $\Sbullet({\shiftby{\Vs}{1}\oplus \shiftby{V}{2}})$ it is not even homogenous.
 \end{rmk}

\begin{rmk}{The reader may have already noticed that for the degree-$(-4)$ element $\varepsilon$ in defining weak Lie 2-bialgebras,  we do not  consider
 terms of the form $ \symmetricproduct^2 \galgebra$,
 $\symmetricproduct^2\theta^*$,
 $(\symmetricproduct^2\theta)\symmetricproduct (\symmetricproduct^2 \galgebrastar)$, $\symmetricproduct^4\theta$, $\symmetricproduct^4 \galgebrastar $, $(\symmetricproduct^2\theta)\symmetricproduct \galgebra$ and $(\symmetricproduct^2 \galgebrastar)\symmetricproduct \theta^*$.
The reason  is explained in a companion paper \cite{CSX}, where it is
 shown that the infinitesimal of a Poisson 2-group corresponds  exactly to  a
 Lie 2-bialgebra.
Incorporating these missing terms would lead to a notion of quasi weak Lie 2-bialgebra, by analogy with the interpretation of quasi-Lie bialgebras of the usual big brackets. We expect further studies in this direction.}
\end{rmk}
\begin{prop}\label{Prop:tt=0impliesll=0andcc=0}
Let $(\thetaalgebradegone, \galgebradegzero, \varepsilon)$ be a weak
Lie 2-bialgebra as in Definition \ref{Defn:Lie2bialgebra}.
Then $(\thetaalgebradegone{{\lon}}\galgebradegzero,\el)$,
where $\el=\derivescovarphi+\derivesbracket+ \derivesaction+\derivesh$,
is a weak Lie 2-algebra, while $(\thetaalgebradegone{{\lon}}\galgebradegzero,c)$,
where $c=\derivescovarphi+\derivescobracket+ \derivescoaction+\derivescoh$,
is a weak Lie 2-coalgebra.
\end{prop}

\begin{proof}
It is easy to see, by examining each component,
that $\Poissonbracket{\varepsilon,\varepsilon}=0$ implies
$\Poissonbracket{\el,\el}=0$ and $\Poissonbracket{c,c}=0$.
Hence $(\thetaalgebradegone{{\lon}}\galgebradegzero,\el)$ is a weak Lie
2-algebra and $(\thetaalgebradegone{{\lon}}\galgebradegzero,c)$ is a
weak Lie 2-coalgebra by Proposition \ref{Prop:Lie2algebraelements}
and Proposition \ref{Prop:Lie2coalgebraelements}.
\end{proof}

\begin{ex}
Assume that $\galgebra$ is a semisimple Lie algebra. Let
$\ipG{\tobefilledin}{\tobefilledin}$ be its Killing form.
Then $ \homotopymap(x,y,z)=
~\plankconstant\ipG{x}{\ba{y}{z}}$, $\forall~~ x,y,z\in\galgebra$,
is  a Lie algebra  $3$-cocycle, where $\plankconstant$ is a constant.
Let $\thetaalgebra=\reals$. Then the trivial map $\reals\to \galgebra$
 together with $\homotopymap$ becomes a weak Lie  2-algebra, called
the string Lie 2-algebra \cite{MR2068522}. More precisely, the
string Lie 2-algebra is as follows:
\begin{enumerate}
\item $\thetaalgebra$  is the abelian Lie algebra $\reals$;
\item $\galgebra$ is a semisimple Lie algebra;
\item $\phimap:~\thetaalgebra\lon \galgebra$ is the trivial map;
\item the action map
$\moduleaction :\galgebra\otimes \thetaalgebra\to \thetaalgebra$
is the trivial map;
\item $\homotopymap :~\wedge^3\galgebra\lon \thetaalgebra$
is given by the map $\plankconstant \ipG{ \cdot ~ }{\ba{\cdot  ~}{\cdot  ~}}$,
 where $\plankconstant$ is a fixed  constant.
\end{enumerate}
Now fix an element $x\in \galgebra$.
We equip a weak Lie 2-coalgebra on $\reals\to \galgebra$ as follows:
\begin{enumerate}
\item $\galgebrastar$ is an  abelian Lie algebra;
\item $\thetaalgebrastar\cong \reals$ is an  abelian Lie algebra;
\item $\phi^*:~\galgebrastar\lon \thetaalgebrastar$ is the trivial map;
\item the $\thetaalgebrastar$-action on $\galgebrastar$ is given by
$ \unit\moduleaction \xi= \adjoint{x}^* \xi, \ \forall~~\xi\in\galgebrastar$;
\item $\cohomotopymap:~\wedge^3\thetaalgebrastar\lon \galgebrastar$ is the
 trivial map.
\end{enumerate}
One can verify directly that these relations indeed
define a weak Lie 2-bialgebra.
\end{ex}

\section{Lie bialgebra crossed modules}

\subsection{Definition}

\begin{defn}
A Lie algebra crossed module consists of a pair of Lie algebras
$\thetaalgebra$ and $\galgebra$, and a linear map
$\phi:~\thetaalgebra\to \galgebra$ such that $\galgebra$ acts on
$\thetaalgebra$ by derivations and satisfies, for all $x,y\in\galgebra$,
$u,v\in\thetaalgebra$,
\begin{itemize}
\item[1)] ${\phi (u)} \moduleaction v=\ba{u}{v}$;
\item[2)] $\phi (  x \moduleaction u)=\ba{ x }{\phi (u)}$,
\end{itemize}
\end{defn}
where $\moduleaction$ denotes the $\galgebra$-action on $\thetaalgebra$.

\begin{rmk}
Note that 1) and 2)  imply that $\phi$ must be  a Lie algebra homomorphism.
\end{rmk}

We write $\crossedmoduletriple{\thetaalgebra}{\phi}{\galgebra}$
to denote a Lie algebra crossed module.
The associated semidirect product Lie algebra
is denoted by $\galgebra\ltimes\thetaalgebra$.

\begin{prop}\label{Lem:ensentialcrossedmoduleLiealgebra}
Given a Lie algebra $\galgebra$, a  $\galgebra$-module
$\thetaalgebra$ and   a linear map $\phi:~\thetaalgebra\lon
\galgebra$ satisfying the following two conditions:
\begin{eqnarray}\label{Eqn:tp1}&\phi (  x \moduleaction u)=\ba{ x }{\phi
(u)},\\\label{Eqn:tp2} &{\phi (u)} \moduleaction v=-{\phi (v)}
\moduleaction  u ,
\end{eqnarray}for all $ ~ u,v\in
\thetaalgebra,  x \in \galgebra$,  there exists a unique Lie
algebra structure on $\thetaalgebra$ such that
$\crossedmoduletriple{\thetaalgebra}{\phi}{\galgebra} $ is a Lie
algebra crossed module.

{In other words, any Lie algebra crossed module underlying $\phi:~\thetaalgebra\lon
\galgebra$ is determined by a Lie algebra $\galgebra$ and a  $\galgebra$-module
$\thetaalgebra$ satisfying   \eqref{Eqn:tp1} -- \eqref{Eqn:tp2}.}
\end{prop}

\begin{proof}
Define the Lie bracket on $\thetaalgebra$ by
$ \ba{u}{v}={\phi (u)} \moduleaction v,\quad \forall~ u,v\in\thetaalgebra .$
The rest of the claim can be easily verified directly.
\end{proof}

{Comparing with the definition of a weak Lie 2-algebra at the beginning of Section \ref{subsec:weakLie2algebras}, we see that   a Lie algebra crossed module is equivalent to a (strict) Lie 2-algebra. }

We are now ready to introduce the following

\begin{defn}\label{Defn:bi-crossedmoduleofLiealgebra}
 A Lie bialgebra crossed module  is  a pair of Lie algebra crossed modules
in duality:
$\crossedmoduletriple{\thetaalgebra}{\phi}{\galgebra} $ and
$\crossedmoduletriple{\galgebrastar}{\phiUprotate
}{\thetaalgebrastar} $ , where $\phiUprotate=-\phi^*$,
are both Lie algebra crossed modules such that
$(\Lietwo{\thetaalgebra}{\galgebra},\Lietwo{\galgebrastar}{\thetaalgebrastar}
 )$ is a Lie bialgebra.
\end{defn}

\begin{prop}
If $(\crossedmoduletriple{\thetaalgebra}{\phi}{\galgebra},
\crossedmoduletriple{\galgebrastar}{\phiUprotate
}{\thetaalgebrastar})$ is a Lie bialgebra crossed module,
so is $(\crossedmoduletriple{\galgebrastar}{\phiUprotate
}{\thetaalgebrastar}, \crossedmoduletriple{\thetaalgebra}{\phi}{\galgebra})$.
\end{prop}

The following proposition justifies our terminology.

\begin{prop}
\label{Prop:liebicrossedmoduleimpliesliebi}
If $(\crossedmoduletriple{\thetaalgebra}{\phi}{\galgebra}$,
$\crossedmoduletriple{\galgebrastar}{\phiUprotate
}{\thetaalgebrastar} )$ is a Lie bialgebra crossed module,
then both pairs $(\thetaalgebra , \thetaalgebrastar)$ and
$(\galgebra, \galgebrastar)$ are Lie bialgebras.
\end{prop}
\begin{proof}
Since $\thetaalgebra$ and $\thetaalgebrastar$ are
Lie subalgebras of $\Lietwo{\thetaalgebra}{\galgebra}$
and $\Lietwo{\galgebrastar}{\thetaalgebrastar}$, respectively,
and
$(\Lietwo{\thetaalgebra}{\galgebra},\Lietwo{\galgebrastar}{\thetaalgebrastar}
 )$ is a Lie bialgebra, it follows that
$(\thetaalgebra,\thetaalgebrastar)$ is a Lie bialgebra.
Similarly,
$(\galgebra,\galgebrastar)$  is also a Lie bialgebra.
\end{proof}

\begin{ex}
We can  construct a Lie bialgebra crossed module from
an ordinary Lie bialgebra as follows.
 Given a Lie bialgebra $(\thetaalgebra,\thetaalgebrastar)$,
consider  the trivial Lie algebra crossed module $\crossedmoduletriple{\thetaalgebra}{\Id}{\thetaalgebra}$,  where  the second $\thetaalgebra$ acts on the first
$\thetaalgebra$ by the adjoint action.
In the mean time,  consider the dual Lie algebra crossed module
$\crossedmoduletriple{\thetaalgebrastar}{-\Id }{ \thetaalgebrastar}$, where
the second
$\thetaalgebrastar$ is equipped with the opposite Lie bracket:
$-\bas{\tobefilledin}{\tobefilledin}$, and the action of the second
$\thetaalgebrastar$ on the first $\thetaalgebrastar$ is by
$ \kappa_2 \moduleaction { \kappa_1 }=  -\bas{\kappa_2}{\kappa_1} ,\quad\forall~
\kappa_1,\kappa_2\in\thetaalgebrastar$.
It is simple to see that $(\crossedmoduletriple{\thetaalgebra}{\Id}{\thetaalgebra}, \crossedmoduletriple{\thetaalgebrastar}{-\Id }{ \thetaalgebrastar})$
is a Lie bialgebra crossed module.
\end{ex}

\subsection{Main theorem}

We are now ready to state our main theorem.

\begin{thm}
\label{Thm:Liebialgebracm1-1strictLie2bialgebra}
There is a bijection between  Lie bialgebra crossed modules
and (strict) Lie 2-bialgebras.
\end{thm}

\begin{rmk} {In fact, the collection of Lie bialgebra crossed modules form a strict 2-category and so does that of Lie 2-bialgebras.   The above theorem can be enhanced to an equivalence of these  2-categories.  This will be investigatedsomewhere else.}
\end{rmk}
We need a few lemmas before proving this theorem.

For $k\geq 1$,  write
\begin{equation}\label{Eqt:Wk}
W_k  =\set{w\in
\galgebra\wedge(\wedge^{k-1}\thetaalgebra)~|~\inserts_{\zeta_1}\inserts_{\phi^*\zeta_2}~w
=- \inserts_{\zeta_2}\inserts_{\phi^*\zeta_1}~w,~ \ \forall \zeta_1, \zeta_2\in
\galgebrastar}.
\end{equation}

Let
$$
\phipush:~\quad \wedge^\bullet(\crossedmodulealgebra)\lon
\wedge^\bullet(\crossedmodulealgebra)
$$
denote  the degree-$0$   derivation with respect to the wedge product
such that $\phipush( x +u)=\phi(u)$,  $\forall  x \in\galgebra$, $u\in\thetaalgebra$.

\begin{lem}
\label{prop:coliealgebracrossedmoduledata}
A Lie algebra crossed module structure  on
 $\crossedmoduletriple{\galgebrastar}{\phiUprotate
}{\thetaalgebrastar} $, where $\phiUprotate=-\phi^*$, is equivalent to
a pair of linear  maps
$\skewdelta:~\galgebra\lon W_2\subset \galgebra\wedge \thetaalgebra$   and $\domega:~\thetaalgebra\lon \wedge^2 \thetaalgebra$ satisfying the following conditions
\begin{itemize}
\item[1)]$\phipush\smalcirc\omega=\skewdelta \smalcirc \phi$;
    \item[2)]$\domega^2=0$;
    \item[3)] $(\domega+\skewdelta)\circ\skewdelta=0$.
\end{itemize}
Here,  we consider both  $\domega$ and $\skewdelta$ as
degree-$1$ derivations on the exterior algebra
$\wedge^\bullet(\crossedmodulealgebra)$,  by letting
$\domega|_{\galgebra}=0$ and  $\skewdelta|_{\thetaalgebra}=0$.

 Moreover in this case,
the cobracket $\partial:~\Lietwo{\thetaalgebra}{\galgebra}\lon
\wedge^2 (\Lietwo{\thetaalgebra}{\galgebra})$ corresponding
to  the Lie algebra structure on
$\Lietwo{\galgebrastar}{\thetaalgebrastar}$ is given by:
\begin{equation}\label{Eqt:partialx+u}
\partial(x+u)=\domega(u)+\skewdelta(x)+\dpi(x),\qquad\forall x\in \galgebra,u\in \thetaalgebra,
\end{equation}
where $\dpi$ is a linear map $~\galgebra\lon \wedge^2 \galgebra$ given by
$ \dpi=-\thalf \phipush \circ \skewdelta$.
\end{lem}

\begin{proof}
According to   Proposition  \ref{Lem:ensentialcrossedmoduleLiealgebra},
a Lie algebra crossed module structure underlying
$\crossedmoduletriple{\galgebrastar}{\phiUprotate
}{\thetaalgebrastar} $ is equivalent  to
assigning a Lie algebra structure $\bas{\tobefilledin}{\tobefilledin}$ on $\thetaalgebrastar$ and an action $\moduleaction$ of $\thetaalgebrastar$ on $\galgebrastar$ such that
\begin{eqnarray}\label{Eqn:tp3}
&\phiUprotate ( \kappa \moduleaction \xi)=\ba{\kappa }{\phiUprotate
(\xi)}_*,\\\label{Eqn:tp4} &{\phiUprotate (\xi)} \moduleaction \zeta=-{\phiUprotate (\zeta)}
\moduleaction  \xi ,
\end{eqnarray}for all $ ~ \xi,\zeta\in
\galgebrastar,~  \kappa \in \thetaalgebrastar$.
 Introduce  linear maps  $\skewdelta$ and $\domega$ by
\begin{eqnarray*}
\pairing{\skewdelta(x)}{\xi\wedge\kappa}&=&\pairing{x}{\kappa\moduleaction \xi};\\
\pairing{\domega(u)}{\kappa_1\wedge \kappa_2}&=& -\pairing{u}{\bas{\kappa_1}{\kappa_2}},
\end{eqnarray*}
$\forall x\in\galgebra$, $u\in \thetaalgebra$, $\xi\in\galgebrastar$, $\kappa, \kappa_1$, $\kappa_2\in \thetaalgebrastar$.
It is simple to see that $\domega^2=0$ is equivalent to
the Jacobi identity for $\bas{\tobefilledin}{\tobefilledin}$,
and $(\domega+\skewdelta)\circ\skewdelta=0$ is equivalent to
that $\moduleaction$ is an action of $\thetaalgebrastar$ on $\galgebrastar$. Moreover, Eq.~\eqref{Eqn:tp3} is equivalent  to the condition $\phipush\smalcirc\omega=\skewdelta \smalcirc \phi$ and Eq.~\eqref{Eqn:tp4} is equivalent to the condition that $\skewdelta$ takes values in $W_2$.

To  prove  Eq.~\eqref{Eqt:partialx+u}, we have,
 $\forall x\in\galgebra$, $u\in \thetaalgebra$, $\xi,\zeta\in\galgebrastar$, $\kappa, \kappa_1$, $\kappa_2\in \thetaalgebrastar$,
\begin{eqnarray*}
\pairing{\partial(u)}{\kappa_1\wedge\kappa_2}&=&
-\pairing{u}{\bas{\kappa_1}{\kappa_2}}=\pairing{\domega(u)}{\kappa_1\wedge \kappa_2},\\
\pairing{\partial(x)}{\kappa\wedge\xi}&=&
-\pairing{x}{\bas{\kappa }{\xi}}=-\pairing{x}{\kappa\moduleaction \xi}=\pairing{\skewdelta(x)}{ \kappa\wedge \xi},
\end{eqnarray*}
and
\begin{eqnarray*}
\pairing{\partial(x)}{\xi\wedge \zeta}&=&-\pairing{x}{\bas{\xi}{\zeta}}\\
&=&-\pairing{x}{ {\phiUprotate(\xi)}\moduleaction {\zeta}}\\
&=& \pairing{\skewdelta(x)}{{\phiUprotate(\xi)}\wedge {\zeta}}\\
&=&\pairing{-\thalf ( \phipush  \circ \skewdelta) (x)}{\xi\wedge \zeta}.
\end{eqnarray*}
The conclusion thus follows.
\end{proof}

\begin{prop}
\label{prop:liebialgebracrossedmoduledata}
Let $\crossedmoduletriple{\thetaalgebra}{\phi}{\galgebra} $ be
 a Lie algebra crossed module.
  It is a  Lie bialgebra crossed module
if and only if there
is a pair of linear maps $(\skewdelta,\domega)$ as in  Lemma
 \ref{prop:coliealgebracrossedmoduledata}
that, in addition,  satisfies the following  conditions:
\begin{itemize}
\item[1)] $\delta$ is a Lie algebra $1$-cocycle;
\item[2)] $x \moduleaction\omega(u)-\omega( x \moduleaction
u)=\mathrm{Pr}_{\wedge^2\thetaalgebra}(\ba{u}{\skewdelta(  x )})$,
for all $ x \in \galgebra$, $u\in\thetaalgebra$.
\end{itemize}
\end{prop}

\begin{proof}
Assume that $\crossedmoduletriple{\thetaalgebra}{\phi}{\galgebra} $ and
$\crossedmoduletriple{\galgebrastar}{\phiUprotate }{\thetaalgebrastar} $ are
  Lie algebra crossed modules.

It suffices to prove that Conditions 1)-2) are equivalent to
$(\Lietwo{\thetaalgebra}{\galgebra},\Lietwo{\galgebrastar}{\thetaalgebrastar}
 )$ being  a Lie bialgebra. The latter is equivalent to
 \begin{equation}\label{Eqt:temp1}
\partial{ \ba{E}{F} }=\ba{E}{\partial(F)}-\ba{F}{\partial(E)},\quad\forall~
E,F\in \Lietwo{\thetaalgebra}{\galgebra},
\end{equation}
 where $\partial: \Lietwo{\thetaalgebra}{\galgebra}
\to \wedge^2 (\Lietwo{\thetaalgebra}{\galgebra})$ is the cobracket
 as given in Eq.~\eqref{Eqt:partialx+u}.

If  both  $E $ and  $F$ are in $ \galgebra$, it is simple to see that
 Eq. \eqref{Eqt:temp1}
is equivalent to  $\skewdelta: \galgebra\to \galgebra\wedge \thetaalgebra$
  being a Lie algebra $1$-cocycle.
On the other hand, we claim that when
  $E=x\in \galgebra$ and  $F=u\in \thetaalgebra$, Eq. \eqref{Eqt:temp1}
is equivalent to the second condition in the statement of the proposition.
First of all, note that Eq. \eqref{Eqt:temp1} implies that
\begin{eqnarray}\nonumber
\partial \ba{x}{u}&=&\ba{x}{\partial(u)}-\ba{u}{\partial(x)}
\label{eqn:temp3} \\
&=&x\moduleaction (\domega (u))-\ba{u}{\skewdelta(x)- \thalf
(\phipush\circ\skewdelta)(x)}.
\end{eqnarray}
Since $\partial \ba{x}{u}=\domega (x\moduleaction u)$, it suffices to
prove that $\ba{u}{\skewdelta(x)-\thalf (\phipush\circ\skewdelta)(x)}=
\mathrm{Pr}_{\wedge^2\thetaalgebra}(\ba{u}{\skewdelta(  x )})$.

For this purpose, let us assume that $\delta(x)= \sum_{i} y_i\wedge v_i$,
where $y_i\in \galgebra$ and  $v_i\in \thetaalgebra$.
The condition  that $\delta(x)\in W_2$ is essentially equivalent to
\begin{equation}\label{Eqt:temp4}
\sum_i \bigl(\pairing{\phi(v_i)}{\xi}y_i+ \pairing{y_i}{\xi}\phi(v_i)\bigr)=0,\quad\forall \xi\in \galgebrastar.
\end{equation}
The latter implies that, for any $u\in \thetaalgebra$,
\[ \sum_i\bigl( (y_i\moduleaction u)\wedge \phi(v_i)-y_i\wedge (\phi(v_i)\moduleaction u)\bigr)=0 .\]
Indeed, for all $  \xi\in \galgebrastar$,
\begin{eqnarray*}
\inserts_{\xi} \Bigl( \sum_i\bigl((y_i\moduleaction u)\wedge \phi(v_i)-y_i\wedge (\phi(v_i)\moduleaction u)\bigr)\Bigr)&=&-\sum_i\bigl( (y_i\moduleaction u)\pairing{ \phi(v_i)}{\xi}+\pairing{y_i}{\xi} (\phi(v_i)\moduleaction u)\bigr)\\
&=&-\sum_i \bigl(\pairing{\phi(v_i)}{\xi}y_i+ \pairing{y_i}{\xi}\phi(v_i)\bigr)\moduleaction u\\
&=&0.
\end{eqnarray*}
Thus, we have
  \begin{eqnarray*}
  \ba{u}{\skewdelta(x)-\thalf (\phipush\circ\skewdelta)(x)}&=&\ba{u}{\sum_{i} y_i\wedge v_i-\thalf \sum_{i} y_i\wedge \phi(v_i)}\\
  &=& \sum_{i}\bigl(\ba{u}{y_i}\wedge v_i
  +y_i\wedge \ba{u}{v_i}+\thalf( y_i\moduleaction u) \wedge \phi(v_i)+\thalf y_i\wedge  ({\phi(v_i)}\moduleaction {u})
  \bigr)\\
  &=&\sum_{i}\ba{u}{y_i}\wedge v_i+\thalf\sum_i\bigl(
  (y_i\moduleaction u)\wedge \phi(v_i)-y_i\wedge (\phi(v_i)\moduleaction u)
  \bigr)\\
  &=&\sum_{i}\ba{u}{y_i}\wedge v_i\\
&=&\mathrm{Pr}_{\wedge^2\thetaalgebra}(\ba{u}{\skewdelta(  x )}).
  \end{eqnarray*}

Finally, if both $E$ and $F$ are in $\thetaalgebra$, Eq. \eqref{Eqt:temp1}
is equivalent to $\omega: \thetaalgebra\to \wedge^2 \thetaalgebra$
being a Lie algebra 1-cocycle. However, the latter follows from
Conditions 1)-2).  To see this, for any $u, v\in \thetaalgebra$,
we have
\begin{eqnarray*}
\domega\ba{u}{v}&=&\domega(\phi(u)\moduleaction v)\\
&=& \phi(u)\moduleaction \domega(v)-\mathrm{Pr}_{\wedge^2\thetaalgebra}(\ba{v}{\skewdelta(  \phi(u) )})\\
&=&
\ba{u}{\domega(v)}-\mathrm{Pr}_{\wedge^2\thetaalgebra}(\ba{v}{\phipush(\omega(u))})\\
&=&\ba{u}{\domega(v)}-\ba{v}{\domega(u)}.
\end{eqnarray*}
Here in the last  equality,  we used the following identity
\[ \mathrm{Pr}_{\wedge^2\thetaalgebra}(\ba{v}{\phipush(\zeta )})=
\ba{v}{\zeta }, \ \ \ \forall \zeta\in \wedge^2 \theta ,\]
which can be  proved by a direct verification.

This concludes the proof.
\end{proof}

Now we  are ready to prove Theorem \ref{Thm:Liebialgebracm1-1strictLie2bialgebra}.

\begin{proof}[Proof of Theorem \ref{Thm:Liebialgebracm1-1strictLie2bialgebra}]
Let  $\varepsilon=  \derivesbracket+
\derivesaction +\derivescovarphi+\derivescobracket+ \derivescoaction \in
 \degreesubspace{\mathcal{S}}{-4}$, where  $\derivesbracket,\derivesaction,  \derivescovarphi,\derivescobracket,\derivescoaction $
are given as in Eqs. \eqref{eqt:derivesdata} and \eqref{eqt:derivescodata}.
 It is simple to see that  the equation $\{\varepsilon, \varepsilon\}=0$
is equivalent to the following three  equations:
\begin{eqnarray}\label{Eqt:tempa}
\Poissonbracket{\derivesbracket+
\derivesaction +\derivescovarphi,\derivesbracket+
\derivesaction +\derivescovarphi}&=&0;\\
\label{Eqt:tempb}
\Poissonbracket{\derivescovarphi+\derivescobracket+
\derivescoaction,\derivescovarphi+\derivescobracket+
\derivescoaction} &=&0;\\
\label{Eqt:tempc}
\Poissonbracket{\derivesbracket,\derivescoaction}
+\Poissonbracket{\derivesaction,\derivescobracket}
+\Poissonbracket{\derivesaction,\derivescoaction}
&=&0.
\end{eqnarray}
According to Proposition \ref{Prop:Lie2algebraelements},
Eq. \eqref{Eqt:tempa} is equivalent to
  that $\crossedmoduletriple{\thetaalgebra}{\phi}{\galgebra}$
 is a Lie algebra crossed module,  where
$\phi(u)=\derivedby{\derivesvarphi}(u),\quad\forall~ u\in \thetaalgebra$,
the Lie bracket on $\galgebra$ is given by
$\ba{x}{y}=\derivedby{\derivesbracket}(x,y),\quad\forall~ x,y\in \galgebra$,
and the $\galgebra$-action   on $\thetaalgebra$ is given by
$x\moduleaction  u=\derivedby{\derivesaction}(x,u),
\quad\forall x\in \galgebra,u\in \thetaalgebra$.
According to Proposition \ref{Prop:Lie2coalgebraelements},
Eq. \eqref{Eqt:tempb} is equivalent to that
$\crossedmoduletriple{\thetaalgebra}{\phi}{\galgebra}$
 is a Lie 2-coalgebra, or
$\crossedmoduletriple{\galgebrastar}{\phiUprotate
}{\thetaalgebrastar} $ is a Lie algebra crossed module.
It is simple to show, by a  straightforward computation,
that the linear maps $\skewdelta$ and $\domega$ associated
to the   Lie algebra crossed module
 $\crossedmoduletriple{\galgebrastar}{\phiUprotate
}{\thetaalgebrastar} $ are related to
 $\derivescobracket$ and $\derivescoaction$ by the following relations:
$$
\pairing{\skewdelta(x)}{\xi\wedge\kappa}=\Poissonbracket{x,\Poissonbracket{\Poissonbracket{\derivescoaction,\kappa},\xi}}=
 -\Poissonbracket{\Poissonbracket{ \derivedby{\derivescoaction}(x),\xi},\kappa};
$$
$$
\pairing{\domega(u)}{\kappa_1\wedge \kappa_2}=
\Poissonbracket{u,\Poissonbracket{\Poissonbracket{\derivescobracket,\kappa_1},\kappa_2}}
=\Poissonbracket{\Poissonbracket{\derivedby{\derivescobracket}(u),\kappa_1},\kappa_2},
$$
$\forall \xi\in\galgebrastar$, $\kappa,\kappa_1,\kappa_2\in\thetaalgebrastar$.

Since  the left hand side of Eq.~\eqref{Eqt:tempc} belongs to
 $(\symmetricproduct^2\galgebrastar)\symmetricproduct \galgebra\symmetricproduct \thetaalgebra+(\symmetricproduct^2\thetaalgebra)\symmetricproduct \galgebrastar\symmetricproduct \thetaalgebrastar$,
hence we have
\begin{eqnarray*}
 &&\Poissonbracket{\Poissonbracket{{\derivedby{\Poissonbracket{\derivesbracket,\derivescoaction} +\Poissonbracket{\derivesaction,\derivescobracket}
+\Poissonbracket{\derivesaction,\derivescoaction}}}(x,y),\xi},\kappa}\\
&=&\Poissonbracket{\Poissonbracket{\derivedby{\Poissonbracket{\derivesbracket,\derivescoaction}}(x,y) +\derivedby{\Poissonbracket{\derivesaction,\derivescoaction}}(x,y),\xi},\kappa}\\
&=&
\Poissonbracket{\Poissonbracket{\derivedby{\derivescoaction}(\ba{x}{y})
+\derivedby{\derivesbracket+\derivesaction}(\derivedby{\derivescoaction}(x),y)
-\derivedby{\derivesbracket+\derivesaction}(\derivedby{\derivescoaction}(y),x),\xi},\kappa}
\\
&=&\pairing{-\skewdelta\ba{x}{y}+\ba{x}{\skewdelta(y)}-\ba{y}{\skewdelta(x)}~}{~\xi\wedge \kappa};
\end{eqnarray*}
and
\begin{eqnarray*}
&&
\Poissonbracket{\Poissonbracket{\derivedby{\Poissonbracket{\derivesbracket,\derivescoaction} +\Poissonbracket{\derivesaction,\derivescobracket}
+\Poissonbracket{\derivesaction,\derivescoaction}}(x,u),\kappa_1},\kappa_2}
\\
&=&\Poissonbracket{\Poissonbracket{\derivedby{\Poissonbracket{\derivesaction,\derivescobracket}}(x,u) +\derivedby{\Poissonbracket{\derivesaction,\derivescoaction}}(x,u),\kappa_1},\kappa_2}\\
&=&\Poissonbracket{\Poissonbracket{\derivedby{\derivescobracket}( {x}\moduleaction {u})
+\derivedby{\derivesaction}(\derivedby{\derivescobracket}(u),x)
-\derivedby{\derivesaction}(\derivedby{\derivescoaction}(x),u),\kappa_1},\kappa_2}\\
&=&\pairing{-x \moduleaction\omega(u)+\omega( x \moduleaction
u)+\mathrm{Pr}_{\wedge^2\thetaalgebra}(\ba{u}{\skewdelta( x )})}{\kappa_1\wedge\kappa_2}.
\end{eqnarray*}

Therefore it follows that Eq. \eqref{Eqt:tempc} is equivalent to
that the pair $(\skewdelta,\domega)$ satisfies the two
compatibility  conditions in
 Proposition \ref{prop:liebialgebracrossedmoduledata}.
Hence we conclude that $ \{t, t\}=0 $ is equivalent to that the couple $\crossedmoduletriple{\thetaalgebra}{\phi}{\galgebra} $ and
$\crossedmoduletriple{\galgebrastar}{\phiUprotate
}{\thetaalgebrastar} $ is a Lie bialgebra crossed module.
\end{proof}

\begin{thm}
\label{Thm:mathchedpairiff}
Assume that $\crossedmoduletriple{\thetaalgebra}{\phi}{\galgebra}$ and
$\crossedmoduletriple{\galgebrastar}{\phiUprotate
}{\thetaalgebrastar}$  are Lie algebra crossed modules.  Then
they form    a Lie bialgebra crossed module if and only if
$(\galgebra,\thetaalgebrastar)$ is a matched pair of Lie algebras,
where the $\galgebra$-action on
$\thetaalgebrastar$ is the dual to the given $\galgebra$-action
on $\thetaalgebra$, while
the $\thetaalgebrastar$-action on $\galgebra$
is dual to the given $\thetaalgebrastar$-action on $\galgebrastar$.
\end{thm}

\begin{proof}
By  definition, $(\galgebra,\thetaalgebrastar)$ is
a matched pair of Lie algebras if and only if
\begin{align}\label{align:temp1}& \kappa \moduleaction \ba{{x}}{{y}}=\ba{{x}}{\kappa \moduleaction
{y}}-\ba{{y}}{\kappa \moduleaction {x}}+({y}\moduleaction
\kappa )\moduleaction {x}-({x}\moduleaction \kappa )\moduleaction {y},
\\
\label{align:temp2}&{x}\moduleaction
\bas{\kappa_1}{\kappa_2}=\bas{\kappa_1}{{x}\moduleaction
\kappa_2}-\bas{\kappa_2}{{x}\moduleaction \kappa_1}+(\kappa_2\moduleaction
{x})\moduleaction \kappa_1-(\kappa_1\moduleaction {x})\moduleaction \kappa_2,
\end{align}
for all ${x},{y}\in\galgebra$, $\kappa_1,\kappa_2\in \thetaalgebrastar$.

We prove that Eq. \eqref{align:temp1} is equivalent to $\skewdelta$
 being a Lie algebra $1$-cocycle, while   Eq. \eqref{align:temp2}
is equivalent to Condition (2) in Proposition \ref{prop:liebialgebracrossedmoduledata}.

Indeed, a tedious computation leads to
\begin{eqnarray*}
&&\pairing{\ddelta{(\ba{{x}}{{y}})}}{ \xi,\kappa }-
\pairing{{{x}}\moduleaction {\ddelta({y})}}{ \xi,\kappa }+
\pairing{{{y}}\moduleaction {\ddelta({x})}}{ \xi,\kappa } \\
&=&\pairing{\xi}{-\kappa \moduleaction \ba{{x}}{{y}}+\ba{{x}}{\kappa \moduleaction
{y}}-\ba{{y}}{\kappa \moduleaction {x}}+({y}\moduleaction
\kappa )\moduleaction {x}-({x}\moduleaction \kappa )\moduleaction {y}},
\end{eqnarray*}
$\forall$  $x\in\galgebra$, $u\in \thetaalgebra$, $ \xi\in \galgebrastar$, $\kappa \in\thetaalgebrastar$,
and 
\begin{eqnarray*}
&&
\pairing{\domega { ({x}\moduleaction u)}}{\kappa_1,\kappa_2}-
\pairing{{x}\moduleaction
\domega(u)}{\kappa_1,\kappa_2}+\pairing{\mathrm{Pr}_{\wedge^2\thetaalgebra}(\ba{u}{\skewdelta(  x )})}{\kappa_1,\kappa_2}\\
&=&\pairing{ u}{{x}\moduleaction
\bas{\kappa_1}{\kappa_2}+\bas{\kappa_2}{{x}\moduleaction
\kappa_1}-\bas{\kappa_1}{ {x}\moduleaction \kappa_2}-(\kappa_2\moduleaction
{x})\moduleaction \kappa_1+(\kappa_1\moduleaction {x})\moduleaction
\kappa_2}.
\end{eqnarray*}
This concludes the proof.
\end{proof}

\begin{cor}
Let $(\thetaalgebra,\thetaalgebrastar)$ be a
 Lie bialgebra.
 Assume that $\abelianideal$
is a subspace of the center   $Z(\thetaalgebra)$ ( i.e.
$\ba{\thetaalgebra}{\abelianideal}=0$) and  $\domega(\abelianideal)
\subset \wedge^2\abelianideal$, where $\domega: \theta\to \wedge^2 \theta$
is  the cobracket on $\thetaalgebra$.
 Then there is an induced Lie bialgebra crossed module structure
underlying $\crossedmoduletriple{\thetaalgebra}{\phi}{ \galgebra}$,
 where $\galgebra= \thetaalgebra/\abelianideal$ is the quotient Lie algebra
 and $\phi$ is the projection.
\end{cor}
\begin{proof}
Identifying  $\galgebrastar$ with $\abelianideal^0\subset \thetaalgebrastar$,
we see that    $\galgebrastar$ is an ideal of $\thetaalgebrastar$,
 and the map $\phiUprotate:~\galgebrastar\lon \thetaalgebrastar$
 is the composition of  the inclusion with $-I$.
Hence  $\crossedmoduletriple{\galgebrastar}{\phiUprotate}{ \thetaalgebrastar}$
is a  Lie algebra crossed module.
To prove that
($\crossedmoduletriple{\thetaalgebra}{\phi}{ \galgebra}$,
$\crossedmoduletriple{\galgebrastar}{\phiUprotate}{ \thetaalgebrastar}$)
 is a Lie bialgebra crossed module, it suffices to
prove that $(\galgebra,\thetaalgebrastar)$ is a matched pair of
Lie algebras according to Theorem \ref{Thm:mathchedpairiff}.
Note that $(\thetaalgebra,\thetaalgebrastar)$ is a Lie algebra matched pair
since it is a  Lie bialgebra. Therefore it remains to show that
 $\abelianideal\oplus 0$ is an ideal of the double Lie algebra
$D=\thetaalgebra\bowtie\thetaalgebrastar$. In fact,
for any $u\in \abelianideal$ and
 $\kappa\in\thetaalgebrastar$, we have
 $\coadjoint{u}\kappa=0$ and $\coadjoint{\kappa}u\in \abelianideal$.
 Hence $[u,\kappa]_D=\coadjoint{u}\kappa-\coadjoint{\kappa}u~\in \abelianideal\oplus 0$.

This concludes the proof.
\end{proof}

\begin{ex}
Consider   the Lie subalgebra $ \mathrm{u}(n)\subset \mathfrak{gl}_n(\CC)$
of $n\times n$ skew-Hermitian matrices.
Let  ${\thetaalgebra}\subset \mathfrak{gl}_n(\CC) $
be the Lie subalgebra consisting  of
upper triangular matrices whose diagonal elements are real numbers.
It is standard
 that $(\thetaalgebra,\mathrm{u}(n))$ is a Lie bialgebra.
   Indeed $\thetaalgebra\oplus\mathrm{u}(n)\cong \mathfrak{gl}_n(\CC)$,
and both $\thetaalgebra$ and $\mathfrak{gl}_n(\CC)$ are lagrangian
subalgebras of $\mathfrak{gl}_n(\CC)$ under the nondegenerate
pairing $\pairing{X}{Y}=\mathrm{Im}(\mathrm{Tr}(XY))$, $\forall$ $X,Y\in \mathfrak{gl}_n(\CC)$.
  Hence $(\mathfrak{gl}_n(\CC),{\thetaalgebra},{\mathrm{u}(n)})$ is a Manin triple and thus
  $(\thetaalgebra,\mathrm{u}(n))$  forms a Lie bialgebra.

Let  $\abelianideal=\reals I$. It is clear that $\abelianideal$
 is the center of $\thetaalgebra$ and  $\domega(\abelianideal)=0$.
Hence $\galgebra=\thetaalgebra/\abelianideal$ can be identified with  the Lie algebra
of traceless upper triangular matrices whose diagonal elements are real.
As a consequence,
 $\crossedmoduletriple{\thetaalgebra}{\phi}{ \galgebra}$,
where $\phi$ is the map $A\to A-\text{tr}A$, is a Lie bialgebra crossed module.
\end{ex}

\end{document}